\documentclass[11pt,reqno]{article}

\usepackage{microtype}

\usepackage{amsmath, amssymb, amsfonts, bm, cite}\usepackage{mathrsfs}
\renewcommand{\mathbf}[1]{\boldsymbol{#1}}

\usepackage{geometry}
\usepackage{fullpage}
\usepackage{pdfsync, graphicx}
\usepackage{amsthm}

\usepackage[colorlinks=true,citecolor=blue,bookmarksnumbered=true,hyperfootnotes=false]{hyperref}

\newtheorem{theorem}{Theorem}[section]

\newtheorem{lemma}[theorem]{Lemma}
\newtheorem{corollary}[theorem]{Corollary}
\newtheorem{definition}{Definition}[section]
\newtheorem{remark}{Remark}[section]
\newtheorem*{remark*}{Remark}

\newcommand{\concept}[1]{\textbf{#1}}

\theoremstyle{plain}
\newtheorem{proposition}[theorem]{Proposition}
\newtheorem{claim}[theorem]{Claim}

\newcommand{\ANN}{\mathsf{ANN}}
\newcommand{\PM}{\mathsf{PM}}

\title{Simple average-case lower bounds for approximate near-neighbor from isoperimetric inequalities}

\author{Yitong Yin\thanks{State Key Lab for Novel Software Technology, Nanjing University, China. \texttt{yinyt@nju.edu.cn}. This work was supported by NSFC grants 61272081 and 61321491.}
}

\date{}

\begin{document}

\maketitle

\begin{abstract}
We prove an $\Omega(d/\log \frac{sw}{nd})$ lower bound for the average-case cell-probe complexity of deterministic or Las Vegas randomized algorithms solving approximate near-neighbor (ANN) problem in $d$-dimensional Hamming space in the cell-probe model with $w$-bit cells, using a table of size $s$. This lower bound matches the highest known worst-case cell-probe lower bounds  for any static data structure problems.

This average-case cell-probe lower bound is proved in a general framework which relates the cell-probe complexity of ANN to isoperimetric inequalities in the underlying metric space. A tighter connection between ANN lower bounds and isoperimetric inequalities is established by a stronger richness lemma proved by cell-sampling techniques.
\end{abstract}

\section{Introduction}\label{section-intro}
The nearest neighbor search ($\mathsf{NNS}$) problem is a fundamental problem in Computer Science. In this problem, a database $y=(y_1,y_2,\ldots,y_n)$ of $n$ points from a metric space $(X, \mathrm{dist})$ is preprocessed to a data structure, and at the query time given a query point $x$ from the same metric space, we are asked to find the point $y_i$  in the database which is closest to $x$ according to the metric.

In this paper, we consider a decision and approximate version of $\mathsf{NNS}$, the approximate near-neighbor ($\ANN$) problem, where the algorithm is asked to distinguish between the two cases:  (1) there is a point in the databases that is  $\lambda$-close to the query point for some radius $\lambda$, or (2) all points in the database are $\gamma\lambda$-far away from the query point, where $\gamma\ge 1$ is the approximation ratio.

The complexity of nearest neighbor search has been extensively studied in the cell-probe model, a classic model for data structures. In this model, the database is encoded to a table consisting of memory cells. Upon each query, a cell-probing algorithm answers the query by making adaptive cell-probes to the table. 
The complexity of the problem is measured by the tradeoff between the time cost (in terms of number of cell-probes to answer a query) and the space cost (in terms of sizes of the table and cells).
There is a substantial body of work on the cell-probe complexity of $\mathsf{NNS}$ for various metric space~\cite{borodin1999lower,chakrabarti1999lower,barkol2000tighter,jayram2003cell,chakrabarti2004optimal, liu2004strong, patrascu06eps2n, patrascu08Linfty, panigrahy2008geometric,panigrahy2010lower, kapralov2012nns, yin2014certificates}.

It is widely believed that $\mathsf{NNS}$ suffers from the ``curse of dimensionality''~\cite{indyk2004nnh}: The problem may become intractable to solve when the dimension of the metric space becomes very high. 
Consider the most important example, $d$-dimensional Hamming space $\{0,1\}^d$ with $d\ge C\log n$ for a sufficiently large constant $C$. The conjecture is that $\mathsf{NNS}$ in this metric remains hard to solve when either approximation or randomization is allowed individually.

In a series of pioneering works~ \cite{borodin1999lower,barkol2000tighter,jayram2003cell,liu2004strong,patrascu06eps2n}, by a rectangle-based technique of asymmetric communication complexity known as the richness lemma~\cite{miltersen1998data}, cell-probe lower bounds in form of $\Omega(d/\log s)$, where $s$ stands for the number of cells in the table, were proved for deterministic approximate near-neighbor (due to Liu~\cite{liu2004strong}) and randomized exact near-neighbor (due to Barkol and Rabani~\cite{barkol2000tighter}). Such lower bound is the highest possible lower bound one can prove in the communication model. 
This fundamental barrier was overcome by an elegant self-reduction technique introduced in the seminal work of~P\v{a}tra\c{s}cu and Thorup~\cite{patrascu2010higher}, in which the cell-probe lower bounds for deterministic $\ANN$ and randomized exact near-neighbor were improved to $\Omega(d/\log \frac{sw}{n})$, where $w$ represents the number of bits in a cell.
More recently, in a previous work of us~\cite{yin2014certificates}, by applying the technique of P\v{a}tra\c{s}cu and Thorup to the certificates in data structures, the lower bound for deterministic $\ANN$ was further improved to $\Omega(d/\log \frac{sw}{nd})$. 
This last lower bound behaves differently for the polynomial space where $sw=\mathrm{poly}(n)$, near-linear space where $sw=n\cdot \mathrm{polylog}(n)$, and linear space where $sw=O(nd)$. In particular, the bound becomes $\Omega(d)$ when the space cost is strictly linear in the entropy of the database, i.e.~when $sw=O(nd)$.

When both randomization and approximation are allowed, the complexity of $\mathsf{NNS}$ is substantially reduced. With polynomial-size tables, a $\Theta(\log\log d/\log\log\log d)$ tight bound was proved for randomized approximate $\mathsf{NNS}$ in $d$-dimensional Hamming space~\cite{chakrabarti1999lower,chakrabarti2004optimal}. If we only consider the decision version, the randomized $\ANN$ can be solved with $O(1)$ cell-probes on a table of polynomial size~\cite{chakrabarti2004optimal}.
For tables of near-linear size, a technique called cell-sampling was introduced by Panigrahy~\textit{et al.}~\cite{panigrahy2008geometric,panigrahy2010lower} to prove $\Omega(\log n/\log\frac{sw}{n})$ lower bounds for randomized $\ANN$. This was later extended to general asymmetric metrics~\cite{abdullah2015directed}.

Among these lower bounds, the randomized $\ANN$ lower bounds of Panigrahy \emph{et al.}~\cite{panigrahy2008geometric,panigrahy2010lower} were proved explicitly for \emph{average-case} cell-probe complexity.
The significance of average-case complexity for $\mathsf{NNS}$ was discussed in their papers. A recent breakthrough in upper bounds~\cite{andoni2015optimal} also attributes to solving the problem on a random database. 
Retrospectively, the randomized exact near-neighbor lower bounds due to the density version of richness lemma~\cite{borodin1999lower,barkol2000tighter,jayram2003cell} also hold for random inputs.
All these average-case lower bounds hold for Monte Carlo randomized algorithms with fixed worst-case cell-probe complexity. 
This leaves open an important case: the average-case cell-probe complexity for the deterministic or Las Vegas randomized algorithms for $\ANN$, where the number of cell-probes may vary for different inputs.

\subsection{Our contributions}

We study the average-case cell-probe complexity of deterministic or Las Vegas randomized algorithms for the approximate near-neighbor ($\ANN$) problem, where the number of cell-probes to answer a query may vary for different query-database pairs and the average is taken with respect to the distribution over input queries and databases.

For $\ANN$ in Hamming space $\{0,1\}^n$, the hard distribution over inputs is very natural: 
Every point $y_i$ in the database $y=(y_1,y_2,\ldots,y_n)$ is sampled uniformly and independently from the Hamming space $\{0,1\}^d$, and the query point $x$ is also a point sampled uniformly and independently from $\{0,1\}^d$. According to earlier average-case lower bounds~\cite{panigrahy2008geometric,panigrahy2010lower} and the recent data-dependent LSH algorthm~\cite{andoni2015optimal}, this input distribution seems to capture the hardest case for nearest neighbor search and is also a central obstacle to overcome for efficient algorithms. 

By a simple proof,  we show the following lower bound for the average-case cell-probe complexity of $\ANN$ in Hamming space with this very natural input distribution.

\begin{theorem}\label{main-ANN-lower-bound}
For $d\ge 32\log n$ and $d<n^{o(1)}$, any deterministic or Las Vegas randomized algorithm solving $(\gamma,\lambda)$-approximate near-neighbor problem  in $d$-dimensional Hamming space in the cell-probe model with $w$-bit cells for $w<n^{o(1)}$, using a table of size $s<2^d$, must have expected cell-probe complexity $t=\Omega\left(\frac{d}{\gamma^2\log\frac{sw\gamma^2}{nd}}\right)$, where the expectation is taken over both the uniform and independent input database and query   and the random bits of the algorithm.
\end{theorem}
This lower bound matches the highest known worst-case cell-probe lower bounds for \emph{any} static data structure problems. Such lower bound was only known for polynomial evaluation~\cite{siegel1989universal,larsen2012higher} and also worst-case deterministic $\ANN$ due to our previous work~\cite{yin2014certificates}.

We also prove an average-case cell-probe lower bound for $\ANN$ under $\ell_\infty$-distance. The lower bound matches the highest known worst-case lower bound for the problem~\cite{patrascu08Linfty}.

In fact, we prove these lower bounds in a unified framework that relates the average-case cell-probe complexity of $\ANN$  to isoperimetric inequalities regarding an expansion property of the metric space.

Inspired by the notions of metric expansion defined in~\cite{panigrahy2010lower}, we define the following notion of expansion for metric space.
Let $(X,\mathrm{dist})$ be a metric space. The $\lambda$-neighborhood of a point $x\in X$, denoted as $N_\lambda(x)$ is the set of all points in $X$ within distance $\lambda$ from $x$. 
Consider a distribution $\mu$ over $X$.
We say 
the $\lambda$-neighborhoods are \concept{weakly independent} under distribution $\mu$, if for any point $x\in X$, the measure of the $\lambda$-neighborhood $\mu(N_\lambda(x))<\frac{\beta}{n}$ for a constant $\beta<1$.
We say the $\lambda$-neighborhoods are \concept{$(\Phi,\Psi)$-expanding} under distribution $\mu$, if for any point set $A\subseteq X$ with $\mu(A)\ge\frac{1}{\Phi}$, we have $\mu(N_\lambda(A))\ge1-\frac{1}{\Psi}$, where $N_{\lambda}(A)$ denotes the set of all points within distance $\lambda$ from some point in $A$.

Consider the database $y=(y_1,y_2,\ldots,y_n)\in X^n$ with every point $y_i$ sampled independently from $\mu$, and the query $x\in X$ sampled independently from $\mu$. We denote this input distribution as $\mu\times\mu^n$.
We prove the following lower bound.
\begin{theorem}\label{main-ANN-expansion}
For a metric space $(X,\mathrm{dist})$, assume the followings:
\begin{itemize}
\item
the $\gamma\lambda$-neighborhoods are weakly independent under distribution $\mu$; 
\item
the $\lambda$-neighborhoods are $(\Phi,\Psi)$-expanding under distribution $\mu$.
\end{itemize}
Then any deterministic or Las Vegas randomized algorithm solving $(\gamma,\lambda)$-approximate near-neighbor problem in $(X,\mathrm{dist})$ in the cell-probe model with $w$-bit cells, using a table of size $s$, must have expected cell-probe complexity 
\[
t=\Omega\left(\frac{\log\Phi}{\log\frac{sw}{n\log\Psi}}\right) \quad\text{ or }\quad t=\Omega\left(\frac{n\log\Psi}{w+\log s}\right)
\] 
under input distribution $\mu\times\mu^n$.
\end{theorem}


The key step to prove such a theorem is a stronger version of the richness lemma that we prove in Section~\ref{section-richness}. The proof of this stronger richness lemma uses an idea called ``cell-sampling'' introduced by Panigrahy \emph{et al.}~\cite{panigrahy2010lower} and later refined by Larsen~\cite{larsen2012higher}. 
This new richness lemma as well as this connection between the rectangle-based techniques (such as the richness lemma) and information-theory-based techniques (such as cell-sampling) are of interests by themselves.

\section{Preliminary}

Let $(X,\mathrm{dist})$ be a metric space. Let $\gamma\ge 1$ and $\lambda\ge 0$. The \concept{$(\gamma,\lambda)$-approximate near-neighbor problem} $(\gamma,\lambda)\text{-}\ANN^n_X$ is defined as follows:  A database $y=(y_1,y_2,\ldots,y_n)\in X^n$ of $n$ points from $X$ is preprocessed and stored as a data structure. Upon each query $x\in X$, by accessing the data structure we want to distinguish between the following two cases: (1) there is a point $y_i$ in the database such that $\mathrm{dist}(x,z)\le\lambda$; (2) for all points $y_i$ in the database we have $\mathrm{dist}(x,z)>\gamma\lambda$. For all other cases the answer can be arbitrary.

More abstractly, given a universe $X$ of queries and a universe $Y$ of all databases, a \concept{data structure problem} is a function $f: X\times Y\to Z$ that maps every pair of \concept{query} $x\in X$ and \concept{database} $y\in Y$  to an answer $f(x,y)\in Z$. In our example of  $(\gamma,\lambda)\text{-}\ANN^n_X$, the query universe is the metric space $X$, the database universe is the set $Y=X^n$ of all tuples of $n$ points from $X$, and $f$ maps each query $x\in X$ and database $y\in Y$ to an Boolean answer: $f(x,y)=0$ if there is a $\lambda$-near neighbor of $x$ in the database $y$; $f(x,y)=1$ if no points in the database $y$ is a $\gamma\lambda$-near neighbor of $x$; and $f(x,y)$ can be arbitrary if otherwise. Note that due to a technical reason, we usually use $1$ to indicate the ``no near-neighbor'' case.

Given a {data structure problem} $f:X\times Y\rightarrow Z$, a code $T:Y\to\Sigma^s$ with alphabet $\Sigma=\{0,1\}^w$ encodes every database $y\in Y$ to a \concept{table} $T_y$ of $s$ \concept{cells} with each cell storing a word of $w$ bits.
We use $[s]=\{1,2,\ldots,s\}$ to denote the set of indices of cells. For each $i\in[s]$, we use $T_y[i]$ to denote the content of the $i$-th cell of table $T_y$; and for $S\subseteq[s]$, we write $T_y[S]=(T_y[i])_{i\in S}$ for the tuple of the contents of the cells in $S$.
Upon each query $x\in X$, a \concept{cell-probing algorithm} adaptive retrieves the contents of the cells in the table $T_y$ (which is called \concept{cell-probes}) and outputs the answer $f(x,y)$ at last. 
Being adaptive means that the cell-probing algorithm is actually a decision tree: In each round of cell-probing the address of the cell to probe next is determined by the query $x$ as well as the contents of the cells probed in previous rounds.
Together, this pair of code and decision tree is called a \concept{cell-probing scheme}.

For randomized cell-probing schemes, the cell-probing algorithm takes a sequence of random bits as its internal random coin. In this paper we consider only deterministic or Las Vegas randomized cell-probing algorithms, therefore the algorithm is guaranteed to output a correct answer when it terminates.

When a cell-probing scheme is fixed, the size $s$ of the table as well as the length $w$ of each cell are fixed. These two parameters together give the space complexity. And the number of cell-probes may vary for each pair of inputs $(x,y)$ or may be a random variable if the algorithm is randomized. Given a distribution $\mathcal{D}$ over $X\times Y$, the \concept{average-case cell-probe complexity} for the cell-probing scheme is given by the expected number of cell-probes to answer $f(\boldsymbol{x},\boldsymbol{y})$ for a $(\boldsymbol{x},\boldsymbol{y})$ sampled from $\mathcal{D}$, where the expectation is taken over both the input distribution $\mathcal{D}$ and the internal random bits of the cell-probing algorithm.

\section{A richness lemma for average-case cell-probe complexity}\label{section-richness}
The richness lemma (or the rectangle method) introduced in~\cite{miltersen1998data} is a classic tool for proving cell-probe lower bounds. A data structure problem $f:X\times Y\rightarrow\{0,1\}$ is a natural communication problem, and a cell-probing scheme can be interpreted as a communication protocol between the cell-probing algorithm and the table, with cell-probes as communications.

Given a distribution $\mathcal{D}$ over $X\times Y$, a data structure problem $f:X\times Y\rightarrow\{0,1\}$ is \concept{$\alpha$-dense} under distribution $\mathcal D$ if $\mathbb{E}_\mathcal{D}[f(\boldsymbol{x},\boldsymbol{y})]\ge\alpha$. 
A combinatorial rectangle $A\times B$ for $A\subseteq X$ and $B\subseteq Y$ is a monochromatic 1-rectangle in $f$ if $f(x,y)=1$ for all $(x,y)\in A \times B$.

The richness lemma states that if a problem $f$ is dense enough (i.e.~being rich in 1's) and is easy to solve by communication, then $f$ contains large monochromatic 1-rectangles. Specifically, if an $\alpha$-dense problem $f$ can be solved by Alice sending $a$ bits and Bob sending $b$ bits in total, then $f$ contains a monochromatic 1-rectangle of size ${\alpha}\cdot{2^{-O(a)}}\times{\alpha}\cdot{2^{-O(a+b)}}$ in the uniform measure. In the cell-probe model with $w$-bit cells, tables of size $s$ and cell-probe complexity $t$, it means the monochromatic 1-rectangle is of size ${\alpha}\cdot{2^{-O(t\log s)}}\times{\alpha}\cdot{2^{-O(t\log s+tw)}}$. The cell-probe lower bounds can then be proved by refuting such large 1-rectangles for specific data structure problems $f$.

We prove the following richness lemma for average-case cell-probe complexity.

\begin{lemma}\label{lemma-richness-average}
Let $\mu,\nu$ be distributions over $X$ and $Y$ respectively, and let $f:X\times Y\rightarrow \{0,1\}$ be $\alpha$-dense under the product distribution $\mu\times\nu$.
If there is a deterministic or randomized Las Vegas cell-probing scheme solving $f$ on a table of $s$ cells, each cell containing $w$ bits, with expected $t$ cell-probes under input distribution $\mu\times\nu$, then for any $\Delta\in\left[{32t}/{\alpha^2}, s\right]$, there is a monochromatic 1-rectangle $A\times B\subseteq X\times Y$ in $f$ such that $\mu(A)\ge{\alpha\cdot \left(\frac{\Delta}{s}\right)^{O(t/\alpha^2)}}$ and $\nu(B)\ge{\alpha \cdot 2^{-O(\Delta\ln\frac{s}{\Delta}+\Delta w)}}$. 
\end{lemma}
Compared to the classic richness lemma, this new lemma has the following advantages:
\begin{itemize}
\item It holds for average-case cell-probe complexity.
\item It gives stronger result even restricted to worst-case complexity. The newly introduced parameter $\Delta$ should not be confused as an overhead caused by the average-case complexity argument, rather, it strengthens the result even for the worst-case lower bounds. When $\Delta=t$ it gives the bound in the classic richness lemma.
\item The lemma claims the existence of a \emph{family} of rectangles parameterized by  $\Delta$, therefore to prove a cell-probe lower bound it is enough to refute any one rectangle from this family. As we will see, this gives us a power to prove the highest lower bounds (even for the worst case) known to any static data structure problems.
\end{itemize}

The proof of this lemma uses an argument called ``cell-sampling'' introduced by Panigrahy~\emph{et al.}~\cite{panigrahy2008geometric, panigrahy2010lower} for approximate nearest neighbor search and later refined by Larsen~\cite{larsen2012higher} for polynomial evaluation. Our proof is greatly influenced by Larsen's approach.

The rest of this section is dedicated to the  proof of this lemma.

\newcommand{\SubHypergraph}{
To prove our richness lemma, we need the following combinatorial lemma for induced sub-hypergraphs with non-uniform edges.
Let $V$ be a finite set. A \emph{hypergraph} $\mathcal{H}$ with vertex set $V$ is a set $\mathcal{H}\subseteq2^{V}$, where each $e\in\mathcal{H}$, called a \emph{hyperedge} or just edge, is a subset $e\subset V$. Note that in our definition, hyperedges are not necessarily having the same size.
Given a subset $\Gamma\subseteq V$ of vertices, the sub-hypergraph of $\mathcal{H}$ induced by $\Gamma$, denoted as $\mathcal{H}_\Gamma$, is defined as $\mathcal{H}_\Gamma=\{e\in\mathcal{H}\mid e\subseteq \Gamma\}$.
The following generic lemma for the size of induced sub-hypergraphs is proved by an easy application of the probabilistic method combined with Jensen's inequality.

\begin{lemma}\label{lemma-subhypergraph}
Let $\mathcal{H}\subseteq2^{[s]}$ be a hypergraph with vertex set $[s]$. Let $\mu$ be a distribution over all hyperedges  $e\in \mathcal{H}$. Assume that $\mathbb{E}_{\boldsymbol{e}\sim\mu}[|\boldsymbol{e}|]\le t$.
For any $4t\le \Delta\le s$, there exists a subset $\Gamma\subseteq[s]$ of size $|\Gamma|=\Delta$ such that the sub-hypergraph $\mathcal{H}_\Gamma$ induced by vertex subset $\Gamma$ have
\[
\mu(\mathcal{H}_\Gamma)\ge\frac{1}{2}\left(\frac{\Delta}{2s}\right)^t.
\]
\end{lemma}
\begin{proof}
We prove this by the probabilistic method. Let $\Gamma\subseteq[s]$ be chosen uniformly among all subsets of size $|\Gamma|=\Delta$. Then a hyperedge $e\in\mathcal{H}$ is contained by $\Gamma$ with probability ${{s-|e|}\choose {\Delta-|e|}}/{s\choose\Delta}$. We denote this probability as $p(e)$. It can be verified that when $|e|\le\frac{\Delta}{2}$, it holds that 
\[
p(e)\ge \left(\frac{\Delta-|e|}{s-|e|}\right)^{|e|}\ge \left(\frac{\Delta}{2s}\right)^{|e|}.
\] 


Recall that $\mathcal{H}_\Gamma=\{e\in \mathcal{H}\mid e\subseteq \Gamma\}$. Each hyperedge $e\in\mathcal{H}$ appears in $\mathcal{H}_\Gamma$ with probability precisely $p(e)$. By linearity of expectation:
\[
\mathbb{E}[\mu(\mathcal{H}_\Gamma)]=\sum_{e\in\mathcal{H}} p(e)\ge\sum_{e\in\mathcal{H}\atop |e|\le\frac{\Delta}{2}}\left(\frac{\Delta}{2s}\right)^{|e|}\mu(e)
\]
Due to Jensen's inequality and the convexity of $\left(\frac{\Delta}{2s}\right)^x$ in $x$, it holds that
\[
\sum_{e\in\mathcal{H}}\left(\frac{\Delta}{2s}\right)^{|e|}\mu(e)\ge \left(\frac{\Delta}{2s}\right)^{\sum_{e\in\mathcal{H}}|e|\mu(e)}=\left(\frac{\Delta}{2s}\right)^{\mathbb{E}_{e\sim\mu}[|e|]}\ge\left(\frac{\Delta}{2s}\right)^{t}.
\]
And since $\mathbb{E}_{e\sim\mu}[|e|]\le t\le\frac{\Delta}{4}$ and by Markov's inequality:
\[
\sum_{e\in\mathcal{H}\atop |e|>\Delta/2}\left(\frac{\Delta}{2s}\right)^{|e|}\mu(e)
\le \left(\frac{\Delta}{2s}\right)^{\Delta/2}\Pr_{e\sim\mu}\left[|e|>\frac{\Delta}{2}\right]
\le \frac{1}{2}\left(\frac{\Delta}{2s}\right)^{2t}.
\]
Therefore, we have
\[
\mathbb{E}[\mu(\mathcal{H}_\Gamma)]\ge \sum_{e\in\mathcal{H}\atop |e|\le\frac{\Delta}{2}}\left(\frac{\Delta}{2s}\right)^{|e|}\mu(e)\ge\left(\frac{\Delta}{2s}\right)^{t}-\frac{1}{2}\left(\frac{\Delta}{2s}\right)^{2t}\ge\frac{1}{2}\left(\frac{\Delta}{2s}\right)^{t}.
\]
By the probabilistic method, there must exists a $\Gamma\subseteq[s]$ of size $\Delta$ such that $\mu(\mathcal{H}_\Gamma)\ge\frac{1}{2}\left(\frac{\Delta}{2s}\right)^{t}$.
\end{proof}
}

\subsection{Proof of the average-case richness lemma (Lemma~\ref{lemma-richness-average})}


By fixing random bits, it is sufficient to consider only deterministic cell-probing algorithms.

The high level idea of the proof is simple.
Fix a table $T_y$. A procedure called the ``cell-sampling procedure'' chooses the subset $\Gamma$ of $\Delta$ many cells that resolve the maximum amount of positive queries. This associates each database $y$ to a string $\omega=(\Gamma,T_y[\Gamma])$, which we call a \concept{certificate},  where $T_y[\Gamma]=(T_y[i])_{i\in\Gamma}$ represent the contents of the cells in $\Gamma$. Due to the nature of the cell-probing algorithm, once the certificate is fixed, the set of queries it can resolve is fixed. We also observe that if the density of 1's in the problem $f$ is $\Omega(1)$, then there is a $\Omega(1)$-fraction of good databases $y$ such that amount of positive queries resolved by the certificate $\omega$ constructed by the cell-sampling procedure is at least an $(\frac{\Delta}{s})^{O(t)}$-fraction of all queries.
On the other hand, since $\omega\in\binom{[s]}{\Delta}\times\{0,1\}^{\Delta w}$ there are at most $\binom{s}{\Delta}2^{\Delta w}=2^{O(\Delta\ln\frac{s}{\Delta}+\Delta w)}$ many certificates $\omega$. Therefore, at least $2^{-O(\Delta\ln\frac{s}{\Delta}+\Delta w)}$-fraction of good databases (which is at least $2^{-O(\Delta\ln\frac{s}{\Delta}+\Delta w)}$-fraction of all databases) are associated with the same $\omega$. Pick this popular certificate $\omega$, the positive queries that $\omega$ resolves together with the good databases that $\omega$ is associated with form the large monochromatic 1-rectangle.

Now we proceed to the formal parts of the proof.
Given a database $y\in Y$, let $X_y^+=\{x\in X\mid f(x,y)=1\}$ denote the set of positive queries on $y$. 
We use $\mu_y^+=\mu_{X_y^+}$ to denote the distribution induced by $\mu$ on $X_y^+$.

Let $P_{xy}\subseteq[s]$ denote the set of cells probed by the algorithm to resolve query $x$ on database $y$. 
Fix a database $y\in Y$. Let $\Gamma\subseteq[s]$ be a subset of cells. We say a query $x\in X$ is resolved by $\Gamma$ if $x$ can be resolved by probing only cells in $\Gamma$ on the table storing database $y$, i.e.~if $P_{xy}\subseteq\Gamma$. We denote by 
\[
X_y^+(\Gamma)=\{x\in X_y^+\mid  P_{xy}\subseteq \Gamma\}
\]
the set of positive queries resolved by $\Gamma$ on database $y$.
Assume two databases $y$ and $y'$ are \emph{indistinguishable} over $\Gamma$: meaning that for the tables $T_y$ and $T_{y'}$ storing $y$ and $y'$ respectively, the cell contents $T_y[i]=T_{y'}[i]$ for all $i\in\Gamma$. Then due to the determinism of the cell-probing algorithm, we have $X_y^+(\Gamma)=X_{y'}^+(\Gamma)$, i.e.~$\Gamma$ resolve the same set of positive queries on both databases.


\paragraph*{The cell-sampling procedure:}
Fix a database $y\in Y$ and any $\Delta\in\left[{32t}/{\alpha^2}, s\right]$.
Suppose we have a \emph{cell-sampling procedure} which does the following:
The {procedure} deterministically\footnote{Being deterministic here means that the chosen set $\Gamma_y^*$ is a function of $y$.} chooses a unique $\Gamma\subseteq[s]$ such that $|\Gamma|=\Delta$ and the measure $\mu(X_y^+(\Gamma))$ of positive queries resolved by $\Gamma$ is maximized (and if there are more than one such $\Gamma$, the procedure chooses an arbitrary one of them). We use $\Gamma_y^*$ to denote this set of cells chosen by the cell-sampling procedure.
We also denote by $X_y^*=X_y^+(\Gamma_y^*)$ the set of positive queries resolved by this chosen set of cells.

On each database $y$, the cell-sampling procedure chooses for us the most informative set $\Gamma$ of cells of size $|\Gamma|=\Delta$ that resolve the maximum amount of positive queries.
We use $\omega_y=(\Gamma_y^*,T_y[\Gamma_y^*])$
to denote the contents (along with addresses) of the cells chosen by the cell-sampling procedure for database $y$. We call such $\omega_y$ a \concept{certificate} chosen by the cell-sampling procedure for $y$.

Let $y$ and $y'$ be two databases.
A simple observation is that if two databases $y$ and $y'$ have the same certificate $\omega_y=\omega_{y'}$ chosen by the cell-sampling procedure, then the respective sets $X_y^*, X_{y'}^*$ of positive queries resolved on the certificate are going to be the same as well.
\begin{proposition}\label{proposition-indistinguishability}
For any databases $y,y'\in Y$, if $\omega_y=\omega_{y'}$ then $X_y^*=X_{y'}^*$.
\end{proposition}

Let $\tau(x,y)=|P(x,y)|$ denote the number of cell-probes to resolve query $x$ on database $y$. 
By the assumption of the lemma, $\mathbb{E}_{\mu\times\nu}[\tau(\boldsymbol{x},\boldsymbol{y})]\le t$ for the inputs $(\boldsymbol{x},\boldsymbol{y})$ sampled from the product distribution $\mu\times\nu$.
We claim that there are many ``good'' columns (databases) with high density of 1's and low average cell-probe costs. 
\begin{claim}\label{claim-good-database}
There is a collection $Y_{\mathsf{good}}\subseteq Y$ of substantial amount of good databases, such that $\nu(Y_{\mathsf{good}})\ge\frac{\alpha}{4}$ and for every $y\in Y_{\mathsf{good}}$, the followings are true:
\begin{itemize}
\item the amount of positive queries is large: $\mu(X_y^+)\ge\frac{\alpha}{2}$;
\item the average cell-probe complexity \emph{among positive queries} is bounded:
\[
\mathbb{E}_{\boldsymbol{x}\sim\mu_y^+}[\tau(\boldsymbol{x},y)]\le\frac{8t}{\alpha^2}.
\]
\end{itemize}
\end{claim}
\begin{proof}
The claim is proved by a series of averaging principles. First consider $Y_{\mathsf{dense}}=\{y\in Y\mid \mu(X_y^+)\ge\frac{\alpha}{2}\}$ the set of databases with at least $\frac{\alpha}{2}$-density of positive queries. By the averaging principle, we have $\nu(Y_{\mathsf{dense}})\ge \alpha/2$. Since $\mathbb{E}[\tau(\boldsymbol{x},\boldsymbol{y})]\ge\nu(Y_{\mathsf{dense}})\mathbb{E}[\tau(\boldsymbol{x},\boldsymbol{y})\mid y\in Y_{\mathsf{dense}}]$, we have $\mathbb{E}_{\mu\times\nu_{\mathsf{dense}}}[\tau(\boldsymbol{x},\boldsymbol{y})]\le \frac{2t}{\alpha}$, where $\nu_{\mathsf{dense}}=\nu_{Y_{\mathsf{dense}}}$ is the distribution induced by $\nu$ on $Y_{\mathsf{dense}}$. We then construct $Y_{\mathsf{good}}\subseteq Y_{\mathsf{dense}}$ as the set of $y\in Y_{\mathsf{dense}}$ with average cell-probe complexity bounded as $\mathbb{E}_{\boldsymbol{x}\sim\mu}[\tau(\boldsymbol{x},y)]\le\frac{4t}{\alpha}$. By Markov inequality $\nu_{\mathsf{dense}}(Y_{\mathsf{good}})\ge\frac{1}{2}$ and hence $\nu(Y_{\mathsf{good}})\ge\frac{\alpha}{4}$. Note that $\mathbb{E}_{\boldsymbol{x}\sim\mu}[\tau(\boldsymbol{x},y)]\ge \mathbb{E}_{\boldsymbol{x}\sim\mu_y^+}[\tau(\boldsymbol{x},y)]\mu(X_y^+)$. We have $\mathbb{E}_{\boldsymbol{x}\sim\mu_y^+}[\tau(\boldsymbol{x},y)]\le\mathbb{E}_{\boldsymbol{x}\sim\mu}[\tau(\boldsymbol{x},y)]/\mu(X_y^+)\le\frac{8t}{\alpha^2}$ for all $y\in Y_{\mathsf{good}}$. 
\end{proof}

For the rest, we consider only these good databases.
Fix any $\Delta\in\left[{32t}/{\alpha^2}, s\right]$. We claim that for every good database $y\in Y_{\mathsf{good}}$, the cell-sampling procedure always picks a subset $\Gamma_y^*\subseteq[s]$ of $\Delta$ many cells, which can resolve a substantial amount of positive queries:
\begin{claim}\label{claim-informative-cells}
For every $y\in Y_{\mathsf{good}}$, it holds that $\mu(X_y^*)\ge\frac{\alpha}{4}\left(\frac{\Delta}{2s}\right)^{8t/\alpha^2}$.
\end{claim}
\begin{proof}
Fix any good database $y\in Y_{\mathsf{good}}$. We only need to prove there exists a $\Gamma\subseteq[s]$ with $|\Gamma|=\Delta$ that resolve positive queries $\mu(X_y^+(\Gamma))\ge\frac{\alpha}{4}\left(\frac{\Delta}{2s}\right)^{8t/\alpha^2}$. The claims follows immediately.

We construct a hypergraph $\mathcal{H}\subseteq 2^{[s]}$ with vertex set $[s]$ as $\mathcal{H}=\{P_{xy}\mid x\in X_y^+\}$, so that each positive queries $x\in X_y^+$ on database $y$ is associated (many-to-one) to a hyperedge $e\in\mathcal{H}$ such that $e=P_{xy}$ is precisely the set of cells probed by the cell-probing algorithm to resolve query $x$ on database $y$.

We also define a measure $\tilde\mu$ over hyperedges $e\in\mathcal{H}$ as the total measure (in $\mu_y^+$) of the positive queries $x$ associated to $e$. Formally, for every $e\in\mathcal{H}$, 
\[
\tilde{\mu}(e)=\sum_{x\in X_y^+:P_{xy}=e}\mu_y^+(x).
\] 
Since $\sum_{e\in\mathcal{H}}\tilde{\mu}(e)=\sum_{x\in X_y^+}\mu_y^+(x)=1$, this $\tilde{\mu}$ is a well-defined probability distribution over hyperedges in $\mathcal{H}$. 
Moreover, recalling that $\tau(x,y)=|P_{xy}|$, the the average size of hyperedges
\[
\mathbb{E}_{\boldsymbol{e}\sim\tilde{\mu}}[|\boldsymbol{e}|]=\mathbb{E}_{\boldsymbol{x}\sim\mu_y^+}[\tau(\boldsymbol{x},y)]\le \frac{8t}{\alpha^2}.
\]
By the probabilistic method (whose proof is in the full paper~\cite{fullpaper}), there must exist a $\Gamma\subseteq[s]$ of size $|\Gamma|=\Delta$, such that the sub-hypergraph $\mathcal{H}_\Gamma$ induced by $\Gamma$ has 
\[
\tilde{\mu}(\mathcal{H}_\Gamma)\ge\frac{1}{2}\left(\frac{\Delta}{2s}\right)^{8t/\alpha^2}.
\]
By our construction of $\mathcal{H}$, the positive queries associated (many-to-one) to the hyperedges in the induced sub-hypergraph $\mathcal{H}_\Gamma=\{P_{xy}\mid x\in X_y^+\wedge P_{xy}\subseteq\Gamma\}$ are precisely those positive queries in $X_y^+(\Gamma)=\{x\in X_y^+\mid P_{xy}\subseteq\Gamma\}$.
Therefore,
\[
\mu_y^+(X_y^+(\Gamma))=\sum_{x\in X_y^+, P_{xy}\subseteq\Gamma}\mu_y^+(x)=\tilde{\mu}(\mathcal{H}_\Gamma)\ge\frac{1}{2}\left(\frac{\Delta}{2s}\right)^{8t/\alpha^2}.
\]
Recall that $\mu(X_y^+)\ge\frac{\alpha}{2}$ for every $y\in Y_{\mathsf{good}}$. And since $X_y^+(\Gamma)\subseteq X_y^+$, we have
\[
\mu(X_y^+(\Gamma))=\mu_y^+(X_y^+(\Gamma))\mu(X_y^+)\ge\frac{\alpha}{4}\left(\frac{\Delta}{2s}\right)^{8t/\alpha^2}.
\]
The claim is proved.
\end{proof}

Recall that the certificate $\omega_y=(\Gamma_y^*,T_y[\Gamma_y^*])$ is constructed by the cell-sampling procedure for database $y$.
For every possible assignment $\omega\in\binom{[s]}{\Delta}\times\{0,1\}^{\Delta w}$ of certificate, let $Y_\omega$ denote the set of good databases $y\in Y_{\mathsf{good}}$ with this certificate $\omega_y=\omega$.
Due to the determinism of the cell-sampling procedure, this classifies the $Y_{\mathsf{good}}$ into at most $\binom{s}{\Delta}2^{\Delta w}$ many disjointed subclasses $Y_\omega$. Recall that $\nu(Y_{\mathsf{good}})\ge\frac{\alpha}{4}$. By the averaging principle, the following proposition is natural.
\begin{proposition}\label{prop-mono-rectangle-cols}
There exists a certificate $\omega\in\binom{[s]}{\Delta}\times\{0,1\}^{\Delta w}$, denoted as $\omega^*$, such that 
\[
\nu(Y_{\omega^*})\ge\frac{\alpha}{4\binom{s}{\Delta}2^{\Delta w}}.
\]
\end{proposition}


On the other hand, fixed any $\omega$, since all databases $y\in Y_\omega$ have the same $\omega_y^*$, by Proposition~\ref{proposition-indistinguishability} they must have the same $X_y^*$. We can abuse the notation and write $X_\omega=X_y^*$ for all $y\in Y_\omega$.



Now we let $A=X_{\omega^*}$ and $B=Y_{\omega^*}$, where $\omega^*$ satisfies Proposition~\ref{prop-mono-rectangle-cols}. Due to Claim~\ref{claim-informative-cells} and Proposition~\ref{prop-mono-rectangle-cols}, we have
\[
\mu(A)\ge\frac{\alpha}{4}\left(\frac{\Delta}{2s}\right)^{8t/\alpha^2}=\alpha\cdot \left(\frac{\Delta}{s}\right)^{O(t/\alpha^2)}
\quad\text{ and }\quad
\nu(B)\ge\frac{\alpha}{4\binom{s}{\Delta}2^{\Delta w}}={\alpha}\cdot{2^{-O\left(\Delta \ln\frac{s}{\Delta}+\Delta w\right)}}.
\]
Note for every $y\in B=Y_{\omega^*}$, the $A=X_{\omega^*}=X_y^+(\Gamma_y^*)$ is a set of positive queries on database $y$, thus $A\times B$ is a monochromatic 1-rectangle in $f$. This finishes the proof of Lemma~\ref{lemma-richness-average}.

\newcommand{\distance}{\mathit{dist}}

\section{Rectangles in conjunction problems}

Many natural data structure problems can be expressed as a conjunction of point-wise relations between the query point and database points.
Consider data structure problem $f:X\times Y\to\{0,1\}$.
Let $Y=\mathcal{Y}^n$, so that each database $y\in Y$ is a tuple $y=(y_1,y_2,\ldots,y_n)$ of $n$ points from $\mathcal{Y}$. A \concept{point-wise function} $g:X\times\mathcal{Y}\to\{0,1\}$ is given.  The data structure problem $f$ is defined as the conjunction of these subproblems:
\[
\forall x\in X, \forall y=(y_1,y_2,\ldots,y_n)\in Y, \quad f(x,y)=\bigwedge_{i=1}^n g(x,y_i),
\]
Many natural data structure problems can be defined in this way, for example:
\begin{itemize}
\item Membership query: $X=\mathcal{Y}$ is a finite domain. The point-wise function $g(\cdot,\cdot)$ is $\neq$ that indicates whether the two points are unequal. 
\item $(\gamma, \lambda)$-approximate near-neighbor $(\gamma,\lambda)\text{-}\ANN^{n}_X$: $X=\mathcal{Y}$ is a metric space with distance $\mathrm{dist}(\cdot,\cdot)$. The point-wise function $g$ is defined as: for $x,z\in X$, $g(x,z)=1$ if $\mathrm{dist}(x,z)>\gamma\lambda$, or $g(x,z)=0$ if $\mathrm{dist}(x,z)\le \lambda$. The function value can arbitrary for all other cases.
\item Partial match $\PM^{d,n}_\Sigma$: $\Sigma$ is an alphabet, $\mathcal{Y}=\Sigma^d$ and $X=(\Sigma\cup\{\star\})^d$. The point-wise function $g$ is defined as: for $x\in X$ and $z\in \mathcal{Y}$, $g(x,z)=1$ if there is an $i\in[d]$ such that $x_i\not\in\{\star,z_i\}$,  or $g(x,z)=0$ if otherwise. 
\end{itemize}

We show that refuting the large rectangles in the point-wise function $g$ can give us lower bounds for the conjunction problem $f$.

Let $\mu,\nu$ be distributions over $X$ and $\mathcal{Y}$ respectively, and let $\nu^n$ be the product distribution on $Y=\mathcal{Y}^n$. Let $g:X\times \mathcal{Y}\rightarrow \{0,1\}$ be a point-wise function and $f:X\times Y\to\{0,1\}$ a data structure problem defined by the conjunction of $g$ as above.
\begin{lemma}\label{lemma-richness-conjunction}
For $f,g,\mu,\nu$ defined as above,
assume that there is a deterministic or randomized Las Vegas cell-probing scheme solving $f$ on a table of $s$ cells, each cell containing $w$ bits, with expected $t$ cell-probes under input distribution $\mu\times\nu^n$.
If the followings are true:
\begin{itemize}
\item the density of 0's in $g$ is at most $\frac{\beta}{n}$ under distribution $\mu\times\nu$ for some constant $\beta<1$;
\item $g$ does not contain monochromatic 1-rectangle of measure at least $\frac{1}{\Phi}\times\frac{1}{\Psi}$ under distribution $\mu\times\nu$;
\end{itemize}
then
\[
\left( \frac{sw}{n\log\Psi}\right)^{O(t)}\ge \Phi\quad\text{ or }\quad t=\Omega\left(\frac{n\log\Psi}{w+\log s}\right).
\]
\end{lemma}
\begin{proof}
By union bound, the density of 0's in $f$ under distribution $\mu\times\nu^n$ is:
\[
\Pr_{\genfrac{}{}{0pt}{}{x\sim\mu}{y=(y_1,\ldots,y_n)\sim\nu^n}}\left[\bigwedge_{i=1}^n g(x,y_i)=0\right]
\le n\cdot \Pr_{\genfrac{}{}{0pt}{}{x\sim\mu}{z\sim\nu}}[g(x,z)=0]\le n\cdot\frac{\beta}{n}=\beta.
\]

By Lemma~\ref{lemma-richness-average}, the $\Omega(1)$-density of 1's in $f$ and the assumption of existing a cell-probing scheme with parameters $s$, $w$ and $t$, altogether imply that for any $4t\le \Delta\le s$, $f$ has a monochromatic 1-rectangle $A\times B$ such that 
\begin{align}
\mu(A)\ge{\left(\frac{\Delta}{s}\right)^{c_1t}}\quad\text{ and }\quad \nu^n(B)\ge2^{-c_2\Delta(\ln\frac{s}{\Delta}+ w)},\label{eq:conjunction-large-rectangles}
\end{align}
for some constants $c_1,c_2>0$ depending only on $\beta$.

Let $C\subset\mathcal{Y}$ be the largest set of columns in $g$ to form a 1-rectangle with $A$. Formally,
\[
C=\{z\in\mathcal{Y}\mid \forall x\in A, g(x,z)=1\}.
\] 
Clearly, for any monochromatic 1-rectangle $A\times D$ in $g$, we must have $D\subseteq C$. 
By definition of $f$  as a conjunction of $g$, it must hold that for all $y=(y_1,y_2,\ldots,y_n)\in B$, none of $y_i\in y$ has $g(x,y_i)=0$ for any $x\in A$, which means $B\subseteq C^n$, and hence
\[
\nu^n(B)\le \nu^n(C^n)=\nu(C)^n.
\]
Recall that $A\times C$ is monochromatic 1-rectangle in $g$. Due to the assumption of the lemma, either $\mu(A)<\frac{1}{\Phi}$ or $\nu(C)<\frac{1}{\Psi}$.
Therefore, either $\mu(A)<\frac{1}{\Phi}$ or $\nu^n(B)<\frac{1}{\Psi^n}$.

We can always choose a $\Delta$ such that $\Delta=O\left(\frac{n\log\Psi}{w}\right)$ and $\Delta=\Omega\left(\frac{n\log\Psi}{w+\log s}\right)$ to satisfy 
\[
2^{-c_2\Delta(\ln\frac{s}{\Delta}+ w)}>\frac{1}{\Psi^n}.
\]
If such $\Delta$ is less than ${32t}/{(1-\beta)^2}$, then we immediately have a lower bound
\[
t=\Omega\left(\frac{n\log\Psi}{w+\log s}\right).
\]
Otherwise, due to~\eqref{eq:conjunction-large-rectangles}, $A\times B$ is monochromatic 1-rectangle in $f$ with $\nu^n(B)>\frac{1}{\Psi^n}$, therefore it must hold that $\mu(A)<\frac{1}{\Phi}$, which by~\eqref{eq:conjunction-large-rectangles} gives us
\[
\frac{1}{\Phi}>\mu(A) \ge{\left(\frac{\Delta}{s}\right)^{O(t)}}={\left(\frac{n\log\Psi}{s w}\right)^{O(t)}},
\]
which gives the lower bound
\[
\left( \frac{sw}{n\log\Psi}\right)^{O(t)}\ge \Phi.
\]
\end{proof}

\section{Isoperimetry and ANN lower bounds}
Given a metric space $X$ with distance $\mathrm{dist}(\cdot,\cdot)$ and $\lambda\ge 0$, we say that two points $x,x'\in X$ are $\lambda$-close if $\mathrm{dist}(x,x')\le\lambda$, and $\lambda$-far if otherwise.
The $\lambda$-neighborhood of a point $x\in X$, denoted by $N_\lambda(x)$, is the set of all points from $X$ which are $\lambda$-close to $x$. Given a point set $A\subseteq X$, we define $N_\lambda(A)=\bigcup_{x\in A}N_{\lambda}(x)$ to be the set of all points which are $\lambda$-close to some point in $A$.

In~\cite{panigrahy2010lower}, a natural notion of metric expansion was introduced.

\begin{definition}[metric expansion~\cite{panigrahy2010lower}]\label{definition-metric-expansion}
Let $X$ be a metric space and $\mu$ a probability distribution over $X$. Fix any radius $\lambda>0$. Define
\[
\Phi(\delta)\triangleq \min_{A\subset X, \mu(A)\le \delta}\frac{\mu(N_\lambda(A))}{\mu(A)}.
\]
The expansion $\Phi$ of the $\lambda$-neighborhoods in $X$ under distribution $\mu$ is defined as the largest $k$ such that for all $\delta\le \frac{1}{2k}$, $\Phi(\delta)\ge k$.
\end{definition}

We now introduce a more refined definition of metric expansion using two parameters $\Phi$ and $\Psi$.
\begin{definition}[$(\Phi,\Psi)$-expanding]\label{expansion-definition}
Let $X$ be a metric space and $\mu$ a probability distribution over $X$.
The $\lambda$-neighborhoods in $X$ are \concept{$(\Phi,\Psi)$-expanding} under distributions $\mu$ if we have $\mu(N_\lambda(A))\ge1-1/\Psi$ for any $A\subseteq X$ that $\mu(A)\ge1/\Phi$.
\end{definition}

The metric expansion defined in~\cite{panigrahy2010lower} is actually a special case of $(\Phi,\Psi)$-expanding:  The expansion of $\lambda$-neighborhoods in a metric space $X$ is $\Phi$ means the $\lambda$-neighborhoods are $(\Phi,2)$-expanding.
The notion of $(\Phi,\Psi)$-expanding allows us to describe a more extremal expanding situation in metric space: The expanding of $\lambda$-neighborhoods does not stop at measure~$1/2$, rather, it can go all the way to be very close to measure~1. This generality may support higher lower bounds for approximate near-neighbor. 

Given a radius $\lambda>0$ and an approximation ratio $\gamma>1$, recall that the $(\gamma,\lambda)$-approximate near neighbor problem $(\gamma,\lambda)\text{-}\mathsf{ANN}_{X}^{n}$ can be defined as a conjunction $f(x,y)=\bigwedge_{i}g(x,y_i)$ of point-wise function $g:X\times X\to\{0,1\}$ where $g(x,z)=0$ if $x$ is $\lambda$-close to $z$; $g(x,z)=1$ if $x$ is $\gamma\lambda$-far from $z$; and $g(x,z)$ is arbitrary for all other cases. Observe that $g$ is actually $(\gamma,\lambda)\text{-}\mathsf{ANN}_{X}^{1}$, the point-to-point version of the $(\gamma,\lambda)$-approximate near neighbor.

The following proposition gives an intrinsic connection between the expansion of metric space and size of monochromatic rectangle in the point-wise near-neighbor relation.
\begin{proposition}
If the $\lambda$-neighborhoods in $X$ are {$(\Phi,\Psi)$-expanding} under distribution $\mu$, then the function $g$ defined as above does not contain a monochromatic 1-rectangle of measure $\ge \frac{1}{\Phi}\times\frac{1.01}{\Psi}$ under distribution $\mu\times\mu$.
\end{proposition}
\begin{proof}
Since the $\lambda$-neighborhoods in $X$ are {$(\Phi,\Psi)$-expanding}, for any $A\subseteq X$ with $\mu(A)\ge \frac{1}{\Phi}$, we have $\mu(N_\lambda(A))\ge 1-\frac{1}{\Psi}$. And by definition of $g$, for any monochromatic $A\times B$, it must hold that $B\cap N_\lambda(A)=\emptyset$, i.e.~$B\subseteq X\setminus N_\lambda(A)$. Therefore, either $\mu(A)<\frac{1}{\Phi}$, or $\mu(B)=1-\mu(N_\lambda(A))\le\frac{1}{\Psi}<\frac{1.01}{\Psi}$.
\end{proof}

The above proposition together with Lemma~\ref{lemma-richness-conjunction} immediately gives us the following corollary which reduces lower bounds for near-neighbor problems to the isoperimetric inequalities.

\begin{corollary}\label{corollary-ANN-expansion}
Let $\mu$ be a distribution over a metric space $X$. Let $\lambda>0$ and $\gamma\ge 1$.
Assume that there is a deterministic or randomized Las Vegas cell-probing scheme solving $(\gamma,\lambda)\text{-}\mathsf{ANN}_{X}^{n}$ on a table of $s$ cells, each cell containing $w$ bits, with expected $t$ cell-probes under input distribution $\mu\times\mu^n$.
If the followings are true:
\begin{itemize}
\item
$\mathbb{E}_{x\sim\mu}\left[\mu(N_{\gamma\lambda}(x))\right]\le\frac{\beta}{n}$ for a constant $\beta<1$;
\item
the $\lambda$-neighborhoods in $X$ are $(\Phi,\Psi)$-expanding under distribution $\mu$;
\end{itemize}
then 
\[
\left( \frac{sw}{n\log\Psi}\right)^{O(t)}\ge \Phi\quad\text{ or }\quad t=\Omega\left(\frac{n\log\Psi}{w+\log s}\right).
\]
\end{corollary}

\begin{remark}
In~\cite{panigrahy2010lower}, a lower bound for $(\gamma,\lambda)\text{-}\mathsf{ANN}^{n}_X$ was proved with the following form:
\[
\left( \frac{swt}{n}\right)^{t}\ge \Phi.
\]
In our Corollary~\ref{corollary-ANN-expansion}, unless the cell-size $w$ is unrealistically large to be comparable to $n$, the corollary always gives the first lower bound 
\[
\left( \frac{sw}{n\log\Psi}\right)^{O(t)}\ge \Phi. 
\]
This strictly improves the lower bound in~\cite{panigrahy2010lower}.
For example, when the metric space is $\left(2^{\Theta(d)},2^{\Theta(d)}\right)$-expanding, this would give us a lower bound $t=\Omega\left(\frac{d}{\log\frac{sw}{nd}}\right)$, which in particular, when the space is linear ($sw=O(nd)$), becomes $t=\Omega(d)$.
\end{remark}

\subsection{Lower bound for ANN in Hamming space}

Let $X=\{0,1\}^d$ be the Hamming space with  Hamming distance $\mathrm{dist}(\cdot,\cdot)$. 
Recall that $N_\lambda(x)$ represents the $\lambda$-neighborhood around $x$, in this case, the Hamming ball of radius $\lambda$ centered at $x$; and for a set $A\subset X$, the $N_\lambda(A)$ is the set of all points within distance $\lambda$ to any point in $A$. 
For any $0\le r\le d$ $B(r)=|N_r(\bar{0})|$ denote the volume of Hamming ball of radius $r$, where $\bar{0}\in\{0,1\}^d$ is the zero vector. Obviously $B(r)=\sum_{k\le r}\binom{d}{k}$.

The following isoperimetric inequality of Harper is well known.
\begin{lemma}[Harper's theorem~\cite{Harper1966385}]
Let $X=\{0,1\}^d$ be the $d$-dimensional Hamming space. For $A\subset X$, let $r$ be such that $|A|\geq B(r)$. Then for every $\lambda>0$, $|N_{\lambda}(A)|\geq B({r+\lambda})$.
\end{lemma}

In words, Hamming balls have the worst vertex expansion. 

For $0<r<\frac{d}{2}$, the following upper bound for the volume of Hamming ball is well known:
\[
2^{(1-o(1))d H(r/d)}\le \binom{d}{r}\le B(r)\le 2^{d H(r/d)},
\]
where $H(x)=-x\log_2 x-(1-x)\log_2(1-x)$ is the Boolean entropy function.

Consider the Hamming $(\gamma,\lambda)$-approximate near-neighbor problem $(\gamma,\lambda)\text{-}\ANN_{X}^{n}$. The hard distribution for this problem is just the uniform and independent distribution: For the database $y=(y_1,y_2,\ldots, y_n)\in X^n$, each database point $y_i$ is sampled uniformly and independently from $X=\{0,1\}^n$; and the query point $x$ is sampled uniformly and independently from $X$.

\begin{theorem}\label{ANN-lower-bound}
Let $d\ge 32\log n$. For any $\gamma\ge 1$, there is a $\lambda>0$ such that if $(\gamma,\lambda)\text{-}\ANN_{X}^{n}$  can be solved by a deterministic or Las Vegas randomized cell-probing scheme on a table of $s$ cells, each cell containing $w$ bits, with expected $t$ cell-probes for uniform and independent database and query, then $t=\Omega\left(\frac{d}{\gamma^2\log\frac{sw\gamma^2}{nd}}\right)$ or $t=\Omega\left(\frac{ nd}{\gamma^2(w+\log s)}\right)$.
\end{theorem}
\begin{proof}
Choose $\lambda$ to satisfy $\gamma\lambda=\frac{d}{2}-\sqrt{2d\ln(2n)}$. 
Let $\mu$ be uniform distribution over $X$. We are going to show:
\begin{itemize}
\item $\mathbb{E}_{x\sim\mu}[\mu(N_{\gamma\lambda}(x))]\le\frac{1}{2n}$;
\item the $\lambda$-neighborhoods in $X$ are $(\Phi,\Psi)$-expanding under distribution $\mu$ for some $\Phi=2^{\Omega(d/\gamma^2)}$ and $\Psi=2^{\Omega(d/\gamma^2)}$.
\end{itemize}
Then the cell-probe lower bounds follows directly from Corollary~\ref{corollary-ANN-expansion}.

First, by the Chernoff bound, $\mu(N_{\gamma\lambda}(x))\le \frac{1}{2n}$ for any point $x\in X$. Thus trivially $\mathbb{E}_{x\sim\mu}[\mu(N_{\gamma\lambda}(x))]\le\frac{1}{2n}$.

On the other hand, for $d\ge 32\log n$ and $n$ being sufficiently large, it holds that $\lambda\ge\frac{d}{4\gamma}$.
Let $r=\frac{d}{2}-\frac{d}{8\gamma}$. And consider any $A\subseteq X$ with $\mu(A)\ge 2^{-(1-H(r/d))d}$. We have $|A|\ge  2^{d H(r/d)}\ge B(r)$. Then by Harper's theorem, 
\[
\mbox{$|N_\lambda(A)|\ge B\left(r+\lambda\right)\ge B\left(\frac{d}{2}+\frac{d}{8\gamma}\right)
\ge 2^d-B\left(\frac{d}{2}-\frac{d}{8\gamma}\right)$} =2^d-B(r)\ge 2^d-2^{d H(r/d)},
\] 
which means $\mu(N_\lambda(A))\ge 1-2^{-(1-H(r/d))d}$. In other words, the $\lambda$-neighborhoods in $X$ are $(\Phi,\Psi)$-expanding under distribution $\mu$ for $\Phi=\Psi=2^{(1-H(r/d))d}$, where $r/d=\frac{1}{2}-\frac{1}{8\gamma}$. Apparently $1-H(\frac{1}{2}-x)=\Theta(x^2)$ for small enough $x>0$. Hence,  $\Phi=\Psi=2^{\Theta(d/\gamma^2)}$.
\end{proof}

\subsection{Lower bound for ANN under L-infinity norm}
Let $\Sigma=\{0,1,\dots,m\}$ and the metric space is $X=\Sigma^d$ with $\ell_\infty$ distance $\mathrm{dist}(x,y)=\left\|x-y\right\|_{\infty}$ for any $x,y\in X$.

Let $\mu$ be the distribution over $X$ as defined in~\cite{patrascu08Linfty}: First define a distribution $\pi$ over $\Sigma$ as $p(i)=2^{-(2\rho)^{i}}$ for all $i>0$ and $\pi(0)=1-\sum_{i>0}\pi(i)$; and then $\mu$ is defined as $\mu(x_1,x_2,\ldots,x_d)=\pi(x_1)\pi(x_2)\ldots\pi(x_d)$.

The following isoperimetric inequality is proved in~\cite{patrascu08Linfty}.
\begin{lemma}[Lemma~9 of \cite{patrascu08Linfty}]\label{isopetrimetric-l-infty}
For any $A\subseteq X$, it holds that $\mu(N_1(A))\geq (\mu(A))^{1/\rho}$.
\end{lemma}

Consider the $(\gamma,\lambda)$-approximate near-neighbor problem $(\gamma,\lambda)\text{-}\ANN_{\ell_\infty}^{n}$ defined in the metric space $X$ under $\ell_\infty$ distance. The hard distribution for this problem is $\mu\times\mu^n$: For the database $y=(y_1,y_2,\ldots, y_n)\in X^n$, each database point $y_i$ is sampled independently according to $\mu$; and the query point $x$ is sampled independently from $X$ according to $\mu$. 
The following lower bound has been proved in~\cite{patrascu08Linfty} and~\cite{kapralov2012nns}.

Fix any $\epsilon>0$ and $0<\delta<\frac{1}{2}$. Assume $\Omega\left(\log^{1+\epsilon}{n}\right)\leq d\leq o(n)$. For $3<c\leq O(\log\log d)$, define $\rho=\frac{1}{2}(\frac{\epsilon}{4}\log d)^{1/c}>10$. 
Now we choose $\gamma=\log_{\rho}\log d$ and $\lambda=1$.

\begin{theorem}\label{ANN-l-infty}
With $d,\gamma,\lambda,\rho$ and the metric space $X$ defined as above, 
if $(\gamma,\lambda)\text{-}\ANN_{\ell_\infty}^{n}$  can be solved by a deterministic or Las Vegas randomized cell-probing scheme on a table of $s$ cells, each cell containing $w\le n^{1-2\delta}$ bits, with expected $t\le \rho$ cell-probes under input distribution $\mu\times\mu^n$, then $sw= n^{\Omega(\rho/t)}$.
\end{theorem}

\begin{proof}
The followings are true 
\begin{itemize}
\item $\mu(N_{\gamma\lambda}(x))=\frac{e^{-\log^{1+\epsilon/3}{n}}}{n}\le\frac{1}{2n}$ for any $x\in X$ (Claim~6 in~\cite{patrascu08Linfty});
\item the $\lambda$-neighborhoods in $X$ are $(n^{\delta\rho},\frac{n^{\delta}}{n^{\delta}-1})$-expanding under distribution $\mu$ for $\Phi=n^{\delta\rho}$ and $\Psi=2^{\Omega(d/\gamma^2)}$.
\end{itemize}
To see the expansion is true, let $\Phi=n^{\delta\rho}$ and $\Psi=\frac{n^{\delta}}{n^{\delta}-1}$. By Lemma~\ref{isopetrimetric-l-infty}, for any set $A\subset X$ with $\mu(A)\geq \Phi$, we have $\mu(N_{\lambda}(A))\geq n^{-\delta}\geq 1-\frac{1}{\Psi}$. This means  $\lambda$-neighborhoods of $\mathcal{M}$ are $(n^{\delta\rho},\frac{n^{\delta}}{n^{\delta}-1} )$-expanding.

Due to Corollary~\ref{corollary-ANN-expansion}, 
either 
$\left(\frac{sw}{n^{1-\delta}}\right)^{O(t)}
\ge n^{\delta\rho}$
or
$t
=\Omega\left(\frac{n^{1-\delta}}{w+\log s}\right)$.
The second bound is always higher with our ranges for $w$ and $t$. The first bound gives $sw=n^{\Omega(\rho/t)}$.
\end{proof}


\begin{thebibliography}{10}

\bibitem{abdullah2015directed}
Amirali Abdullah and Suresh Venkatasubramanian.
\newblock A directed isoperimetric inequality with application to bregman near
  neighbor lower bounds.
\newblock In {\em STOC'15}.

\bibitem{patrascu08Linfty}
Alexandr Andoni, Dorian Croitoru, and Mihai P{\v a}tra{\c s}cu.
\newblock Hardness of nearest neighbor under L-infinity.
\newblock In {\em FOCS'08}.

\bibitem{patrascu06eps2n}
Alexandr Andoni, Piotr Indyk, and Mihai P{\v a}tra{\c s}cu.
\newblock On the optimality of the dimensionality reduction method.
\newblock In {\em FOCS'06}.

\bibitem{andoni2015optimal}
Alexandr Andoni and Ilya Razenshteyn.
\newblock Optimal data-dependent hashing for approximate near neighbors.
\newblock In {\em STOC'15}.

\bibitem{barkol2000tighter}
Omer Barkol and Yuval Rabani.
\newblock Tighter lower bounds for nearest neighbor search and related problems
  in the cell probe model.
\newblock {\em Journal of Computer and System Sciences}, 64(4):873--896, 2002. Conference version in {\em STOC'00}.

\bibitem{borodin1999lower}
Allan Borodin, Rafail Ostrovsky, and Yuval Rabani.
\newblock Lower bounds for high dimensional nearest neighbor search and related
  problems.
\newblock In {\em Discrete and Computational Geometry}, pages 253--274, 2003.
Conference version in {\em STOC'99}.

\bibitem{chakrabarti1999lower}
Amit Chakrabarti, Bernard Chazelle, Benjamin Gum, and Alexey Lvov.
\newblock A lower bound on the complexity of approximate nearest-neighbor
  searching on the hamming cube.
\newblock  In {\em Discrete and Computational Geometry}, pages 313--328, 2003.
Conference version in {\em STOC'99}.

\bibitem{chakrabarti2004optimal}
Amit Chakrabarti and Oded Regev.
\newblock An optimal randomised cell probe lower bound for approximate nearest
  neighbour searching.
\newblock In {\em SIAM Journal on Computing}, 39(5):1919--1940,2010. Conference version in {\em FOCS'04}.

\bibitem{Harper1966385}
L.H. Harper.
\newblock Optimal numberings and isoperimetric problems on graphs.
\newblock {\em Journal of Combinatorial Theory}, 1(3):385 -- 393, 1966.

\bibitem{indyk2004nnh}
Piotr Indyk.
\newblock Nearest neighbors in high-dimensional spaces.
\newblock {\em Handbook of Discrete and Computational Geometry}, pages
  877--892, 2004.

\bibitem{jayram2003cell}
T.S.~Jayram, Subhash Khot, Ravi Kumar, and Yuval Rabani.
\newblock Cell-probe lower bounds for the partial match problem.
\newblock In {\em Journal of Computer and System Sciences}, 69(3):435--447, 2004.
Conference version in {\em STOC'03}.

\bibitem{kapralov2012nns}
Michael Kapralov and Rina Panigrahy. 
\newblock NNS lower bounds via metric expansion for $\ell_\infty$ and EMD. 
\newblock In {\em ICALP'12}.

\bibitem{larsen2012higher}
Kasper~Green Larsen.
\newblock Higher cell probe lower bounds for evaluating polynomials.
\newblock In {\em FOCS'12}.

\bibitem{liu2004strong}
Ding Liu.
\newblock A strong lower bound for approximate nearest neighbor searching.
\newblock {\em Information Processing Letters}, 92(1):23--29, 2004.

\bibitem{miltersen1998data}
Peter~Bro Miltersen, Noam Nisan, Shmuel Safra, and Avi Wigderson.
\newblock On data structures and asymmetric communication complexity.
\newblock {\em Journal of Computer and System Sciences}, 57(1):37--49, 1998. Conference version in {\em STOC'95}.

\bibitem{panigrahy2008geometric}
Rina Panigrahy, Kunal Talwar, and Udi Wieder.
\newblock A geometric approach to lower bounds for approximate near-neighbor
  search and partial match.
\newblock In {\em FOCS'08}.

\bibitem{panigrahy2010lower}
Rina Panigrahy, Kunal Talwar, and Udi Wieder.
\newblock Lower bounds on near neighbor search via metric expansion.
\newblock In {\em FOCS'10}.

\bibitem{patrascu2010higher}
Mihai P{\v a}tra{\c s}cu and Mikkel Thorup.
\newblock Higher lower bounds for near-neighbor and further rich problems.
\newblock {\em SIAM Journal on Computing}, 39(2):730--741, 2010.
\newblock Conference version in {\em FOCS'06}.

\bibitem{siegel1989universal}
Alan Siegel.
\newblock On universal classes of fast high performance hash functions, their
  time-space tradeoff, and their applications.
\newblock In {\em FOCS'89}.

\bibitem{yin2014certificates}
Yaoyu Wang and Yitong Yin.
\newblock Certificates in data structures.
\newblock In {\em ICALP'14}.

\bibitem{fullpaper}
Yitong Yin.
\newblock Simple average-case lower bounds for approximate near-neighbor from isoperimetric inequalities. 
\newblock {\em arXiv preprint} arXiv:1602.05391.

\end{thebibliography}
\end{document}